\documentclass[onecolumn, 11pt, letterpaper]{article}

\usepackage[linesnumbered,ruled]{algorithm2e}
\usepackage{amsmath}
\usepackage{amsthm}
\usepackage{appendix}
\usepackage[margin = 1in]{geometry}
\usepackage{hyperref}       % hyperlinks
\usepackage{url}            % simple URL typesetting

\usepackage{color}

\newtheorem{definition}{Definition}
\newtheorem{lemma}{Lemma}
\newtheorem{observation}{Observation}
\newtheorem{theorem}{Theorem}

\title{Cake Cutting on Graphs: A Discrete and Bounded Proportional Protocol}

\author{
	Xiaohui Bei \thanks{School of Physical and Mathematical Sciences, Nanyang Technological University. xhbei@ntu.edu.sg.}
	\and
	Xiaoming Sun \thanks{CAS Key Lab of Network Data Science and Technology, Institute of Computing Technology, Chinese Academy of Sciences. sunxiaoming@ict.ac.cn.}
	\and
	Hao Wu \thanks{University of Chinese Academy of Sciences. wuhao164@mails.ucas.ac.cn.}
    \and
    Jialin Zhang \thanks{CAS Key Lab of Network Data Science and Technology, Institute of Computing Technology, Chinese Academy of Sciences. zhangjialin@ict.ac.cn.}
    \and
    Zhijie Zhang \thanks{CAS Key Lab of Network Data Science and Technology, Institute of Computing Technology, Chinese Academy of Sciences. zhangzhijie@ict.ac.cn.}
    \and
    Wei Zi \thanks{University of Chinese Academy of Sciences. ziwei16@mails.ucas.ac.cn.}
}

\date{\vspace{-3ex}}

\begin{document}

	\maketitle

	\begin{abstract}
		The classical cake cutting problem studies how to find fair allocations of a heterogeneous and divisible resource among multiple agents.
		Two of the most commonly studied fairness concepts in cake cutting are \emph{proportionality} and \emph{envy-freeness}.
		It is well known that a proportional allocation among $n$ agents can be found efficiently via simple protocols~\cite{EvenP84-Note}.
		For envy-freeness, in a recent breakthrough, Aziz and Mackenzie~\cite{AzizM16-Discrete-Any} proposed a discrete and bounded envy-free protocol for any number of players. However, the protocol suffers from high multiple-exponential query complexity and it remains open to find simpler and more efficient envy-free protocols.

		In this paper we consider a variation of the cake cutting problem by assuming an underlying graph over the agents whose edges describe their acquaintance relationships, and agents evaluate their shares relatively to those of their neighbors.
		An allocation is called \emph{locally proportional} if each agent thinks she receives at least the average value over her neighbors.
		Local proportionality generalizes proportionality and is in an interesting middle ground between proportionality and envy-freeness: its existence is guaranteed by that of an envy-free allocation, but no simple protocol is known to produce such a locally proportional allocation for general graphs.
		Previous works showed locally proportional protocols for special classes of graphs, and it is listed in both~\cite{AbebeKP17-Fair} and~\cite{BeiQZ17-Networked} as an open question to design simple locally proportional protocols for more general classes of graphs.
		In this paper we completely resolved this open question by presenting a discrete and bounded locally proportional protocol for any given graph.
%		Our protocol borrows ideas from~\cite{AzizM16-Discrete-Any} but also utilizes the graph structure to simplify most of its steps. As a result,
		Our protocol has a query complexity of only single exponential, which is significantly smaller than the six towers of $n$ query complexity of the envy-free protocol given in~\cite{AzizM16-Discrete-Any}.
	\end{abstract}


    \newpage

	\section{Introduction}
	\label{Section: introduction}
	
	The problem of fair division studies how to allocate a set of scarce resources to a set of interested agents in a fair manner.
	When the resource is heterogeneous and divisible, the problem is known as ``cake cutting'' and has a long and intriguing history in multiple disciplines such as economics, social science, political science, and computer science~\cite{BramsT-Fair, RobertsonW98-Cake,Procaccia16-Cake}.
	In the standard model, a single heterogeneous resource, also known as a cake, is represented by the interval $[0,1]$. Each agent has a valuation function which defines her preference over different parts of the resource. The goal is to distribute the resource to the agents using standard queries~\cite{RobertsonW98-Cake}, such that everyone feels she is treated ``fairly''.
	Two of the most prominent concepts to measure fairness in cake cutting are \emph{proportionality} and \emph{envy-freeness}. Informally, proportionality means that each agent, in her own view, gets at least an average share of the cake; and envy-freeness means that each agent weakly prefers her own piece to any other agent's. It is not hard to see that envy-freeness implies proportionality but not vice versa.
	
	It is well known that a proportional allocation among $n$ agents can be efficiently found using $O(n\log{n})$ queries~\cite{EvenP84-Note}. In the meanwhile, since the 1940s, the envy-free cake cutting  problem has baffled the great minds from multiple disciplines.
	% In 1995, Brams and Taylor~\cite{BramsT95-Envy-Free} presented an envy-free protocol for any number of players, but the protocol may require an unbounded number of queries even for four players.
	It was not until 2016 that discrete and bounded envy-free protocols for four and more players were finally proposed by Aziz and Mackenzie~\cite{AzizM16-Discrete-Four,AzizM16-Discrete-Any}. Despite these groundbreaking advances, the main drawback of the new protocols is their high query complexity in the form of six towers of $n$. The unrealistic number of evaluations and cuts required by the protocols prevents them from being put into practical use.
	
	A common assumption made by all previous works is that each agent is aware of which parts every other agent gets, and is comparing her own piece to everyone else's. This is not necessarily the case in many contexts where each agent's knowledge and focus are \emph{local}.
	In light of this, two independent works~\cite{AbebeKP17-Fair,BeiQZ17-Networked} considered a notion of \emph{local fairness}
	\footnote{This notion is termed \emph{local fairness} in~\cite{AbebeKP17-Fair} and \emph{networked fairness} in~\cite{BeiQZ17-Networked}. We adopt the former terminology in this paper. For better clarification, we will sometimes also call the original definitions of these two fairness notions \emph{global proportionality} and \emph{global envy-freeness}.}.
	Their model assumes an underlying social network over the agents, and defines proportionality and envy-freeness locally in relation to their neighbors: given an underlying graph of the agents, an allocation is \emph{locally envy-free} if no agents envies any of her neighbor's share, and is \emph{locally proportional} if every agent  values her own share no less than the average among her neighbors.
	Both local proportionality and local envy-freeness are weaker than the global envy-freeness. Yet from a practical point of view they capture many situations in which global knowledge is unavailable or unrealistic.
	It is also hopeful that by considering only local comparisons, there could exist simpler and more intuitive protocols that produce locally fair allocations.
	
	For local envy-freeness, a nontrivial continuous protocol for trees was proposed in~\cite{BeiQZ17-Networked} when the protocol is allowed to perform the so-called Austin Cut Procedure~\cite{Austin82-Sharing}.
	% Both works listed the design of locally envy-free protocols for more types of graphs as an interesting open question.
	However, this protocol is hard to generalize. What's worse, it seems a far-fetched task to have a simple but general locally envy-free protocol for any given graph. If we have such a protocol, then by applying the result to the complete graph, it would imply a simple globally envy-free protocol. From previous works and attempts, this seems to be a very difficult task.
	
	On the other hand, local proportionality poses an intriguing case. First of all, a simple globally proportional protocol can be easily found~\cite{EvenP84-Note}. However, unlike local envy-freeness, local proportionality is not a monotone property. That is, a globally proportional allocation is not necessarily locally proportional for a particular graph.
	In the meanwhile, for any graph, the existence of a locally proportional allocation is guaranteed from that of a globally envy-free allocation. In~\cite{BeiQZ17-Networked}, the authors gave a simple locally proportional protocol for a special class of graphs known as the descendant graph. It is listed both in~\cite{AbebeKP17-Fair} and~\cite{BeiQZ17-Networked} as an open question to design \emph{simple} and locally proportional protocols for more classes of graphs.
	% In this direction, the ultimate goal would be to give a simple and locally proportional protocol for any given graph.

	\paragraph{Our Contribution.} In this paper, we completely resolve the aforementioned open question by presenting a protocol that is able to produce a locally proportional allocation for \emph{any} given graph.
	The protocol only has a \emph{single exponential} query complexity, which is a significant improvement over the six towers of $n$ query complexity of the globally envy-free protocol from~\cite{AzizM16-Discrete-Any}. 
	It is a \emph{discrete} protocol in that it only requires standard queries~\cite{RobertsonW98-Cake} and does not require any continuous operations such as the moving-knife procedure.
	It also shows that proportionality can indeed be algorithmically generalized to incorporate graphical structures in a simple and efficient way.
	
	Similar to most envy-free cake cutting protocols \cite{AmanatidisC18-Improved, AzizM16-Discrete-Any, AzizM16-Discrete-Four, BramsT95-Envy-Free}, our protocol maintains a locally proportional partial allocation and an unallocated \emph{residue} of the cake.
	It repeatedly updates both parts throughout its execution, until the unallocated residue is allocated \emph{completely}.
	% We use the ``Core Protocol'' introduced in~\cite{AzizM16-Discrete-Any} to partially allocate a residue (i.e., unallocated cake) with one agent as the cutter that results in an envy-free partial allocation among all agents.
	We rely on the concept of \emph{dominant agents} to help eliminate the residue.
	Intuitively speaking, an agent is \emph{dominant} if she remains proportional even if the residue were to be allocated among her neighbors.
	This, as a result, allows us to allocate the residue only among non-dominant agents.
	Surprisingly, in doing so, the remaining non-dominant agents can be made dominant as well in an efficient way,
	based on the observation that the neighbors of a dominant agent are able to collect enough \emph{bonuses} over her and finally \emph{dominate} her by excluding her from the remaining allocation.
	
	Despite its useful corollaries, the argument above does not tell us how the first dominant agent can be created.
	Indeed, this issue is the most difficult part of the problem.
	We manage to resolve it by adopting the idea of exchanging pieces allocated to agents.
	In doing so, we hope that for a specific agent, the pieces of insignificant values in her own view will be transferred to her neighbors' hands.
	However, in general, exchanging pieces may incur large losses in agents' shares.
	As a result, the new allocation after exchange may not maintain the desired fairness any more.
	Previously, Aziz and Mackenzie \cite{AzizM16-Discrete-Any} managed to overcome this difficulty by reserving an enough fraction of the residue and allocating it carefully to cover any possible losses.
	Their method is extremely complicated, resulting in the towers of $n$ query complexity.
	In contrast, we always split the allocation into two parts, and restrict the exchange to take place within one of the two parts, such that the losses incurred are covered by the bonuses reserved in the other part.
	We note that our method is possible since we merely need to achieve the weaker proportionality constraint.

	\paragraph{Related Work.}
	It is well known that a globally proportional allocation among $n$ agents can be efficiently found using $O(n \log n)$ queries \cite{EvenP84-Note}.
	On the other hand, after a series of attempts \cite{Magdon-IsmailBK03-Cake, WoegingerS07-Complexity}, a matching lower bound was finally proved in \cite{EdmondsP06-Cake}.
	In contrast, although a globally envy-free allocation is guaranteed to exist \cite{Stromquist80-How, Su99-Rental}, \emph{finding} such an allocation seems to be extremely hard.
	In 1995, Brams and Taylor proposed a \emph{discrete but unbounded} protocol \cite{BramsT95-Envy-Free}.
	However, their protocol suffers from a significant drawback that the number of queries can be made arbitrarily large by choosing certain agents' valuation functions.
	Later, by resorting to the \emph{moving-knife} procedure, several \emph{continuous} protocols were proposed for four agents \cite{BramsT97-Moving-Knife, BarbanelB04-Cake} and five agents \cite{SaberiW09-Cutting}.
	It was not until 2016 that discrete and bounded protocols for four and more agents were finally proposed by Aziz and Mackenzie \cite{AzizM16-Discrete-Four, AzizM16-Discrete-Any}.
	The only known nontrivial lower bound is $\Omega(n^2)$ \cite{Procaccia09-Thou}.

	The notion of local fairness with respect to social networks was first introduced to the cake cutting problem independently in~\cite{AbebeKP17-Fair} and~\cite{BeiQZ17-Networked}.
	Previously known algorithmic results are mainly from \cite{BeiQZ17-Networked}, which gave a locally envy-free continuous protocol for trees as well as a locally proportional protocol for a special class of graphs known as the descendant graph.
	% The protocol uses approximately $n!$ queries in the worst case.

	The idea of restricted comparison has also been considered in the allocation of indivisible resources.
	For instance, both Chevaleyre et al.~\cite{ChevaleyreEM07} and Brederect et al.~\cite{Bredereck0N18} considered the \emph{graph envy-freeness} where envy can only arise between adjacent agents.
	Todo et al.~\cite{TodoLHMIY11} studied envy-freeness between two groups of agents.
	Aziz et al.~\cite{AzizBCGL18} proposed a family of epistemic notions of fairness based on a social graph, establishing a hierarchy among various notions of fairness.
    Chen and Shah~\cite{chen2017ignorance} defined the so-called \emph{Bayesian envy-freeness} under a Bayesian scenario.

	\paragraph{Structure.}
	In Section \ref{Section: preliminaries}, we give a formal description of the problem and introduce some central concepts as well as the Core Protocol in \cite{AzizM16-Discrete-Any} used in our protocol.
	In section \ref{Section: the proportional protocol}, we present a locally proportional protocol for any given graph as well as its analysis.
%	Specifically, the sub-protocol presented in \ref{Subsection: create the first dominant agent} makes an specified agent dominant.
%	The sub-protocol presented in \ref{Subsection: make every agent dominant} makes all agents dominant, based on the outputs of the previous sub-protocol.
    As a supplement, in section \ref{Secton: toward envy-freeness on trees}, we discretize the continuous protocol in \cite{BeiQZ17-Networked} to obtain a locally envy-free \emph{partial} allocation on trees.
	We conclude our paper by proposing several open problems in Section \ref{Section: conclusion}.

	\section{Preliminaries}
	\label{Section: preliminaries}

	A cake is to be allocated among a set of $n$ agents on an \emph{undirected} social graph $G = (V, E)$, with each vertex $v \in V$ identified with an agent.
	The cake is represented by the interval $[0, 1]$.
	A \emph{piece} of cake refers to a finite union of disjoint subintervals of $[0, 1]$.
	The collection $\{A_v\}_{v\in V}$ of disjoint pieces is called an \emph{allocation} if piece $A_v$ is allocated to agent $v$.
	An allocation is \emph{partial} if the union of the allocated pieces is not the whole cake; otherwise it is \emph{complete}.
	The unallocated part with respect to a partial allocation is called a \emph{residue}, often denoted by $R$.
	In a standard setting, the cake is required to be allocated \emph{completely}.

	Each agent $v \in V$ has a valuation function $f_v$ over different pieces of the cake, which is defined on subintervals of $[0, 1]$ and assumes the following properties.
	\begin{itemize}
		\item \emph{Normalized} and \emph{nonnegative}. $f_v([0, 1]) = 1$ and $f_v(I) \geq 0$ for interval $I \subseteq [0, 1]$.
		\item \emph{Additive}. $f_v(I \cup I') = f_v(I) + f_v(I')$ for disjoint intervals $I, I'$.
		\item \emph{Divisible}. For any interval $I$ and $0 \leq \lambda \leq 1$, there exists an interval $I' \subseteq I$ such that $f_v(I') = \lambda \cdot f_v(I)$.
	\end{itemize}
    We adopt the standard Robertson-Webb query model \cite{RobertsonW98-Cake} to access agents' valuations, in which two kinds of queries are allowed.

    \begin{itemize}
    	\item \emph{Evaluation queries.} Given an interval $[x, y]$ and an agent $v$, the query returns $f_v([x, y])$.
    	\item \emph{Cut queries.} Given a point $x \in [0, 1]$, a value $\alpha$ and an agent $v$, the query returns a point $y$ such that $f_v([x, y]) = \alpha$.
    \end{itemize}
    The complexity of a cake cutting protocol is measured by the total number of queries it uses.

	Let $N(v)$ denote the set of agent $v$'s neighbors and $d_v = |N(v)|$ denote its degree.
	Below are two concepts used to measure the fairness of an allocation on a graph.
	By their definitions, we may assume w.l.o.g~that the graph is connected, since otherwise we can allocate the cake within a particular component of the graph.

	\begin{definition}[Local proportionality]
		\label{Definition: proporitonality}
		An allocation $\{A_v\}_{v\in V}$ is locally proportional on a graph $G = (V, E)$ if for each $v \in V$, $f_v(A_v) \geq \frac{1}{d_v} \sum_{u \in N(v)} f_v(A_u)$.
	\end{definition}

    \begin{definition}[Local envy-freeness]
    	\label{Definition: envy-freeness}
    	An allocation $\{A_v\}_{v\in V}$ is locally envy-free on a graph $G = (V, E)$ if for each $v \in V$ and $u \in N(v)$, $f_v(A_v) \geq f_v(A_u)$.
    \end{definition}

	\subsection{Central Concepts Used in Our Protocol}
	\label{Subsection: central concepts used in our protocol}

	Similar to most cake cutting protocols \cite{AmanatidisC18-Improved, AzizM16-Discrete-Any, AzizM16-Discrete-Four, BramsT95-Envy-Free}, our protocol maintains a locally proportional partial allocation and a residue of the cake.
	It repeatedly invokes the Core Protocol in \cite{AzizM16-Discrete-Any} to diminish the residue throughout its execution, until the residue is allocated \emph{completely}.
	We introduce several concepts below to further identify the quality of an allocation.
	They play central roles in our protocol for eliminating the unallocated residue.

	\begin{definition}[Dominant agents]
		\label{Definition: dominant agents}
		Given an allocation $\{A_v\}_{v\in V}$ and the residue $R$, an agent $v$ is dominant if $f_v(A_v) \geq \frac{1}{d_v}(\sum_{u\in N(v)}f_v(A_u) + f_v(R))$.
	\end{definition}
%   agent $v$ dominates agent $u$ (they are not necessarily adjacent) if $f_v(A_v) \geq f_v(A_u) + f_v(R)$.

%	Intuitively, $v$ dominating $u$ implies that $v$ will not envy $u$ even if the whole residue is allocated to $u$.
    Intuitively, a dominant agent will remain proportional even if the residue were to be allocated among her neighbors.
    As a result, a dominant agent can be safely excluded from the allocation of the residue, reducing the problem to a smaller scale.
    In essence, our protocol is all about how agents can be made dominant.
    We rely on the following concept and observation to create a dominant agent.

    \begin{definition}[Dominance between a pair of agents]
    	\label{Definition: dominance between a pair of agents}
    	Given an allocation $\{A_v\}_{v\in V}$ and the residue $R$, we say agent $v$ dominates agent $u$ if $f_v(A_v) \geq f_v(A_u) + f_v(R)$.
    \end{definition}

	\begin{observation}
		\label{Observation: dominance}
		Given an allocation $\{A_v\}_{v \in V}$ and the residue $R$, if agent $v$ is envy-free and additionally she dominates at least one of her neighbors, then she is dominant.
	\end{observation}

    \begin{proof}
    	Assume that agent $v$ dominates her neighbor $w$, then $f_v(A_v) \geq f_v(A_w) + f_v(R)$.
    	Since agent $v$ is also envy-free, $f_v(A_v) \geq f_v(A_u)$ for $u \in N(v)$.
    	Hence $f_v(A_v) \geq \frac{1}{d_v}(\sum_{u\in N(v)}f_v(A_u) + f_v(R))$, i.e~agent $v$ is dominant.
    \end{proof}

    The next concept is heavily used in our protocol as well as its analysis.
    It also provides an alternative perspective of the concepts above, as discussed in the following paragraph.

    \begin{definition}[Bonus]
    	\label{Definition: bonus}
    	Given an allocation $\{A_v\}_{v \in V}$, the difference $f_v(A_v) - f_v(A_u)$ is called agent $v$'s bonus over agent $u$; the summation $\sum_{u \in N(v)} (f_v(A_v) - f_v(A_u))$ is called agent $v$'s bonus over her neighbors.
    \end{definition}

    It is easy to see that agent $v$ does not envy agent $u$ if her bonus over agent $u$ is nonnegative.
    Agent $v$ is proportional if her bonus over her neighbors is nonnegative.
    Besides, agent $v$ is dominant if her bonus over her neighbors is at least the value of the residue.
    Agent $v$ dominates agent $u$ if her bonus over agent $u$ is at least the value of the residue.

	\subsection{The Core Protocol}
	\label{Subsection: the core protocol}

	Our protocol repeatedly invokes the Core Protocol, which was proposed by Aziz and Mackenzie in their
	celebrated paper \cite{AzizM16-Discrete-Any} to deal with the canonical case where the underlying graph is a complete graph.
	We denote this protocol by Core$(r, N, R)$, where $R$ is the residue to be allocated among agents in $N \subseteq V$, and agent $r \in N$ serves as a \emph{cutter}, who will cut the residue $R$ into $m = |N|$ equally preferable pieces at the very beginning of the protocol.
	The protocol returns a (possibly partial) allocation $\{X_v\}_{v\in V}$ of $R$ and an updated residue $R' = R \backslash (\cup_{v \in V} X_v)$ with the following properties.

	\begin{itemize}
		\item
		Each agent $v \in N$ obtains (possibly a part of) exactly one out of the $m$ pieces cut by the cutter $r$.
		Besides, the cutter $r$ and one more agent $u \neq r$ each obtains a complete piece.
		\item
		The allocation $\{X_v\}_{v \in V}$ is globally envy-free.
		That is, for any $v,u\in N$, $f_v(X_v) \geq f_v(X_u)$, despite whether or not $v$ and $u$ are adjacent.
	\end{itemize}
    It was shown in \cite{AzizM16-Discrete-Any} that the Core Protocol is bounded by $n^{2n + 3}$ queries.
    In Subsection \ref{Subsection: the complexity of core protocol}, we give a tighter analysis and improve the bound to $O(c^n)$ queries for any constant $c > 4$.
    For completeness and readability, we also present the overall description of the Core Protocol in that subsection.

	\begin{definition}[Snapshot]
		\label{Definition: snapshot}
		The allocation $\{X_v\}_{v \in V}$ returned by an execution of the Core Protocol is called a snapshot.
	\end{definition}

    The Core Protocol reduces to the famous Cut and Choose Protocol when $|N| = 2$ and therefore will allocate the residue $R$ completely among agents in $N$.
    When $|N| \geq 3$, the updated residue $R'$ is possibly non-empty.
    Fortunately, as described below, it is guaranteed that in a snapshot the cutter $r$ has a large bonus over some agent $u \in N$, which makes it possible to make the cutter $r$ dominant.

    \begin{definition}[Insignificant piece]
    	\label{Definition: insignificant piece}
    	Given a snapshot $\{X_v\}_{v \in V}$ and the updated residue $R'$ returned by Core$(r, N, R)$ with $m = |N| \geq 3$, a piece $X_u$ is called an insignificant piece if $f_r(X_r) - f_r(X_u) \geq f_r(R') / (m - 2)$.
    	We say the cutter $r$ has a significant bonus over the agent who is allocated an insignificant piece in the snapshot.
    \end{definition}

    \begin{observation}
    	\label{Observation: existence of an insignificant piece}
    	Given a snapshot $\{X_v\}_{v \in V}$ and the updated residue $R'$ returned by Core$(r, N, R)$ with $m = |N| \geq 3$, there exists at least one insignificant piece.
    \end{observation}

    \begin{proof}
    	By the properties of the Core Protocol, in the snapshot the cutter $r$ and some other agent each obtains a piece of value $f_r(R) / m$.
    	Hence the remaining $m - 2$ agents are allocated pieces of value $f_r(R) - \frac{2}{m}f_r(R) - f_r(R')$.
    	By an averaging argument, there exists an agent $u$ who obtains a piece $X_u$ such that
    	\[f_r(X_u) \leq \frac{f_r(R) - \frac{2}{m}f_r(R) - f_r(R')}{m - 2} = \frac{f_r(R)}{m} - \frac{f_r(R')}{m - 2} = f_r(X_r) - \frac{f_r(R')}{m - 2}.\]
    	Therefore, $X_u$ is an insignificant piece.
    \end{proof}

	It turns out that the cutter $r$ can finally dominate the agent $u$ who holds an insignificant piece by the following procedure: starting from $R'$, run the Core Protocol another $k \geq m \ln (m - 2) /2$ times on the continually updated residue with agent $r$ as the cutter.
	On the one hand, the cutter $r$'s bonus over agent $u$ in the final allocation is at least $f_r(R') / (m - 2)$, since the snapshots generated by the following executions of the Core Protocol is envy-free.
	On the other hand, since one execution of the Core Protocol reduces the residue to at most an $(m - 2) / m$ fraction in the cutter's valuation, the final residue is of value at most $(\frac{m-2}{m})^k f_r(R')$, which is smaller than $f_r(R') / (m - 2)$.
	As a result, the cutter $r$ dominates agent $u$ at the moment.

	\section{The Proportional Protocol}
	\label{Section: the proportional protocol}

	Our protocol runs in two stages.
	In the first stage, the protocol specifies an agent $r$ that is from the center\footnote{The center of a graph is the set of vertices whose eccentricities equal to the radius of the graph. Recall that the eccentricity of a vertex is the greatest distance between other vertices and it.} of $G$.
	Then the CreateDominance Protocol in Subsection \ref{Subsection: create the first dominant agent} will guarantee that agent $r$ is dominant.
	In the second stage, starting from agent $r$, we manage to make every other agent dominant by invoking the DiffuseDominance Protocol designed in Subsection \ref{Subsection: make every agent dominant}.
	The remaining residue will also be allocated in this stage, resulting in a proportional allocation of the whole cake.
	We summarize these two stages as the Main Protocol (Algorithm \ref{Algorithm : main protocol}).

	\begin{algorithm}
		\SetAlgoLined
		\KwIn{A graph $G = (V, E)$ with $|V| = n$, a cake $R = [0, 1]$.}
		\KwOut{A proportional allocation of the cake $R$ on graph $G$.}

		Pick an agent $r$ that lies in the center of $G$.

		CreateDominance$(G, R, r)$.

		DiffuseDominance$(G, R, r)$.

		\caption{Main Protocol}
		\label{Algorithm : main protocol}
	\end{algorithm}

	\subsection{Create the First Dominant Agent}
	\label{Subsection: create the first dominant agent}

	The CreateDominance Protocol (Algorithm \ref{Algorithm: create dominance}) depicts how a specified agent $r$ can be made dominant.
	By Observation \ref{Observation: dominance}, it suffices for the task if we are able to find an allocation in which agent $r$ is envy-free and additionally she dominates one of her neighbors.
	To this end, we resort to the Core Protocol and manipulate the snapshots returned by running it multiple times on the continuously updated residue.
	By the definition of insignificant pieces (Definition \ref{Definition: insignificant piece}), agent $r$ will finally dominate one of her neighbors if this neighbor is allocated an insignificant piece in some snapshot.
	Unfortunately, however, it could be the case that the insignificant pieces are always allocated to agent $r$'s non-neighbors.
	Consequently, the most difficult part of the task lies in how it is guaranteed that an insignificant piece is allocated to some neighbor of agent $r$.
	We manage to overcome this difficulty by exchanging agents' pieces allocated in the same snapshot, in the meanwhile maintaining the proportionality of the overall allocation.
	Our key observation is that a pair of \emph{adjacent} agents is able to \emph{safely} exchange their pieces in one of at most $2n$ snapshots.
	Specifically, two adjacent agents will exchange their pieces allocated in the snapshot where their bonuses over each other are both smaller than those in the remaining snapshots.
	In doing so, the loss incurred by exchanging pieces in the selected snapshot will be bounded by the bonuses reserved in the other snapshots, leading the overall allocation to remain proportional.
	As an immediate application of the observation, our protocol manages to transfer the insignificant pieces held by an agent $u$ to some neighbor of agent $r$ along a shortest path from agent $u$ to agent $r$.
	The parameter $C'$ at the beginning of the protocol denotes the number of snapshots required to ensure that some agent $u$ holds an enough number of insignificant pieces such that at least one of them can be successfully transferred to some neighbor of agent $r$.
	The value of $C'$ will be determined in Lemma \ref{Lemma: the complexity of the CreateDominance Protocol}.

	\begin{algorithm}
		\SetAlgoLined
		\KwIn{A graph $G=(V,E)$ with $|V|=n$, a cake $R$, and an agent $r$.}
		\KwOut{A (possibly partial) proportional allocation of $R$ and a corresponding residue, where $r$ is dominant.}

		\For{$j = 1, \cdots, C'$}{
			Run Core$(r, V, R)$ to obtain a snapshot.

			\If{there is a neighbor $v$ of agent $r$ such that agent $v$ obtains an insignificant piece in this snapshot}{
				Run Core$(r, V, R)$ $\frac{n}{2}\ln (n-2)$ times to make agent $r$ dominate agent $v$.

				\Return the current allocation of $R$ and an updated residue.
			}
		}

		Pick an agent $u$ who obtains insignificant pieces in at least $C = C' / (n - 1)$ snapshots.
		Denote by $\mathcal{S}_1 = \{S^j : j = 1, \cdots, C\}$ the collection of these snapshots.

		Find a shortest path $P = u_1 u_2 \cdots u_k r \, (u_1 = u)$ from agent $u$ to agent $r$ (here $k \geq 2$).

		\For{$i = 1, \cdots, k - 1$}{
			Agent $u_i$ chooses a collection $\mathcal{S}'_i$ of $\frac{d_{u_i}|\mathcal{S}_i|}{d_{u_i} + d_{u_{i + 1}} + 1}$ snapshots from $\mathcal{S}_i$ for which she values her bonus over agent $u_{i+1}$ the most.

			Agent $u_{i+1}$ chooses a collection $\mathcal{S}''_i$ of $\frac{d_{u_{i + 1}}|\mathcal{S}_i|}{d_{u_i} + d_{u_{i + 1}} + 1}$ snapshots from $\mathcal{S}_i \backslash \mathcal{S}'_i$ for which she values her bonus over agent $u_i$ the most.

			Let $\mathcal{S}_{i + 1} = \mathcal{S}_i \backslash (\mathcal{S}'_i \cup \mathcal{S}''_i)$.

			For each snapshot in $\mathcal{S}_{i + 1}$, exchange pieces held by agent $u_i$ and agent $u_{i+1}$.
		}

		Run Core$(r, V, R)$ $\frac{n}{2}\ln (n-2)$ times to make agent $r$ dominate agent $u_k$.

		\Return a (possibly partial) allocation of $R$ and an updated residue.

		\caption{CreateDominance}
		\label{Algorithm: create dominance}
	\end{algorithm}

	\begin{lemma}
		\label{Lemma: the analysis of the CreateDominance Protocol}
		Algorithm \ref{Algorithm: create dominance} returns a (possibly partial) locally proportional allocation on graph $G$ where agent $r$ is dominant.
	\end{lemma}

	\begin{proof}
		We assume w.l.o.g.~that the insignificant pieces in the $C'$ snapshots are always allocated to agent $r$'s non-neighbors, since otherwise by Observation \ref{Observation: dominance}, agent $r$ will be dominant after step 4.
		Under this assumption, some non-neighbor $u$ of agent $r$ is guaranteed to obtain the insignificant pieces in at least $C = C' / (n - 1)$ snapshots by an averaging argument.
		Below we show that every agent remains proportional in the overall allocation throughout the remaining part of the protocol.

		First observe that agents \emph{not} in the path $P' = u_1 u_2 \cdots u_k$ remain envy-free and thus proportional throughout the protocol, since they never have their pieces exchanged, and a snapshot is globally envy-free, implying that they always hold the best piece in a snapshot.

		Next we show by induction that each agent in $P'$ remains proportional after agent $u_i$ and agent $u_{i + 1}$ exchanged their pieces in snapshots from $\mathcal{S}_{i + 1}$, assuming that each agent is proportional before the exchange.
		First observe that for any $j > i + 1$, agent $u_j$ will remain envy-free and thus proportional, since she has not yet had her pieces exchanged.
		Next, we consider agent $u_{i + 1}$ and agent $u_i$.
		Below we use $b_j(u, v)$ to denote agent $u$'s bonus over agent $v$ in a given snapshot $S^j$ \emph{before} the exchange and $b'_j(u, v)$ to denote the corresponding bonus \emph{after} the exchange.

		For agent $u_{i + 1}$, we show that her loss incurred by exchanging pieces in snapshots from $\mathcal{S}_{i + 1}$ is covered by her bonuses over agent $u_i$ in snapshots from $\mathcal{S}''_i$, thus remaining proportional.
		Given a snapshot $S^j$ from $\mathcal{S}_{i + 1}$, the bonuses $b_j(u_{i + 1}, v)$ before the exchange are non-negative, since agent $u_{i + 1}$ has not had her piece exchanged.
		After the exchange, the updated bonuses $b'_j$ satisfy that $b'_j(u_{i + 1}, u_i) = - b_j(u_{i + 1}, u_i)$ and $b'_j(u_{i + 1}, v) = b_j(u_{i + 1}, v) - b_j(u_{i + 1}, u_i)$ for agent $v \neq u_i$.
		Consequently, the loss in agent $u_{i + 1}$'s bonus over her neighbors incurred by exchanging pieces in $S^j$ is $d_{u_{i + 1}} \cdot b_j(u_{i + 1}, u_i)$.
		The total loss incurred by exchanging pieces in snapshots from $\mathcal{S}_{i + 1}$ is therefore $d_{u_{i + 1}} \cdot \sum_{S^j \in \mathcal{S}_{i + 1}} b_j(u_{i + 1}, u_i)$.
		By the choice of $\mathcal{S}''_i$, however, $|\mathcal{S}''_i| = d_{u_{i + 1}} \cdot |\mathcal{S}_{i + 1}|$ and agent $u_{i + 1}$'s bonus over agent $u_i$ in any snapshot from $\mathcal{S}''_i$ is at least that in any snapshot from $\mathcal{S}_{i + 1}$.
		Therefore, the loss has been totally covered.

		The same conclusion holds for agent $u_i$, but it deserves to illustrate some additional details.
		Note that the bonuses $b_j(u_i, v)$ in a given snapshot $S_j$ before the exchange could be negative.
		They are indeed the losses incurred when agent $u_{i -1}$ and agent $u_i$ exchanged their pieces and have already been covered at that time.
		Hence there is no need to worry about the first part of the bonuses $b'_j(u_i, v) = b_j(u_i, v) - b_j(u_i, u_{i  + 1})$ after the exchange.
		The second part is bounded by the same argument as that for agent $u_{i + 1}$.

		For any $j < i$, the proportionality of agent $u_j$ follows immediately from the argument below.
		By the analysis above, the reason why an agent remains proportional is that whenever she is reallocated an inferior piece in a snapshot, the loss is covered as if she is reallocated a piece of the same value as the best piece in this snapshot.
		As a result, this agent will still remain proportional as long as she is not reallocated another inferior piece in the future.

		Agent $r$ remains envy-free throughout the protocol.
		Let $C'$ be the value determined in Lemma \ref{Lemma: the complexity of the CreateDominance Protocol} such that $\mathcal{S}_k$ is non-empty, the neighbor $u_k$ of agent $r$ is guaranteed to obtain the insignificant pieces in the snapshots from $\mathcal{S}_k$.
		By Observation \ref{Observation: dominance}, agent $r$ will be dominant after step 16.

	\end{proof}

    \begin{lemma}
    	\label{Lemma: the complexity of the CreateDominance Protocol}
    	Algorithm \ref{Algorithm: create dominance} uses $O(n \cdot \left( \frac{6n}{k} \right) ^ k T(n))$ queries by setting $C = \left(\frac{6n}{k}\right)^k$, where $k$ is the length of the shortest path $P$ found by Algorithm \ref{Algorithm: create dominance} and $T(n)$ is the number of queries used in the Core Protocol.
    \end{lemma}

    \begin{proof}
    	We bound the complexity of the protocol by determining an appropriate value of $C$ such that $\mathcal{S}_k$ is non-empty.
    	Since $|\mathcal{S}_i| = (d_{u_i} + d_{u_{i + 1}} + 1) \cdot |\mathcal{S}_{i + 1}|$, a lower bound of $C$ is
    	\[C = |\mathcal{S}_1| \geq \prod_{i = 1}^{k - 1}(d_{u_i} + d_{u_{i + 1}} + 1).\]
    	For $i = 1, 2, 3$, let $k_i = |\{u_j \in P \mid j \equiv i \pmod 3 \}|$.
    	Thus $k_1 + k_2 + k_3 = k$.
    	Since $P$ is a shortest path, agent $u_i$ and agent $u_j$ have no common neighbors if $|i - j| > 2$.
    	Therefore, for $i = 1, 2, 3$,
    	\[d_{u_i} + d_{u_{i + 3}} + \cdots + d_{u_{3(k_i - 1) + i}} + k_i \leq n.\]
    	Summing up these inequalities yields
    	\[\sum_{i=1}^{k}d_{u_i} + k \leq 3n.\]
    	By the fundamental inequality,
    	\[\prod_{i = 1}^{k - 1}(d_{u_i} + d_{u_{i + 1}} + 1) \leq \left(\frac{2\sum_{i=1}^{k}d_{u_i} + k}{k}\right)^k \leq \left(\frac{6n}{k}\right)^k.\]
    	Hence it suffices to set $C = \left(\frac{6n}{k}\right)^k$.
    	The lemma follows since the complexity of the Core Protocol is dominant in Algorithm \ref{Algorithm: create dominance} and it is invoked $O(nC)$ times.
    \end{proof}

    We summarize the above results as Theorem \ref{Theorem: the analysis of the CreateDominance Protocol}.

    \begin{theorem}
    	\label{Theorem: the analysis of the CreateDominance Protocol}
    	Algorithm \ref{Algorithm: create dominance} returns a (possibly partial) locally proportional allocation on graph $G$ where agent $r$ is dominant and uses $O(n \cdot \left( \frac{6n}{k} \right) ^ k T(n))$ queries,  where $k$ is the length of the shortest path $P$ found by Algorithm \ref{Algorithm: create dominance} and $T(n)$ is the number of queries used in the Core Protocol.
    \end{theorem}

	\subsection{Make Every Agent Dominant}
	\label{Subsection: make every agent dominant}

	In principle, each agent can be made dominant by running the CreateDominance Protocol with it as an input.
	However, the DiffuseDominance Protocol (Algorithm \ref{Algorithm: diffuse dominance}) shows that it is possible to fulfill the task by only using additional $O(n^2)$ calls of the Core Protocol once the first dominant agent $r$ was created.
	The key observation is that a dominant agent will remain proportional even if she is excluded from the allocation of the residue.
	Thus by invoking the Core Protocol to allocate the residue among the remaining agents, the neighbors of a dominant agent are able to collect enough bonuses over her and finally dominate her.

	\begin{algorithm}
		\SetAlgoLined
		\KwIn{A graph $G=(V,E)$ with $|V|=n$, the residue $R$,  a dominant agent $r \in V$.}
		\KwOut{A proportional allocation of the cake $R$ on graph $G$.}
		$D = \{r\}$.

		\While{$D \neq V$}{
			Pick a neighbor $v$ of $D$.

			Run Core$(v, V \backslash D, R)$ $K_v$ times, where $K_v$ will be determined in the analysis.

			$D = D \cup \{v\}$.
		}

		\Return a proportional allocation of the cake $R$ on graph $G$.

		\caption{DiffuseDominance}
		\label{Algorithm: diffuse dominance}
	\end{algorithm}

	\begin{theorem}
		\label{Theorem: the analysis of the DiffuseDominance Protocol}
		Algorithm \ref{Algorithm: diffuse dominance} returns a locally proportional complete allocation of the cake $R$ on graph $G$ with $O(n^2)$ calls of the Core Protocol.
	\end{theorem}

	\begin{proof}
		We first fix some notations for a particular `while' loop where a neighbor $v$ of $D$ serves as a cutter of the Core Protocol.
		Let $\{A_u^0\}_{u \in V}$ be the partial allocation and $R^0$ be the residue just before this loop.
		For $k = 1, \cdots, K_v$, let $\{X_u^k\}_{u \in V}$ be the snapshot and $R^k = R^{k - 1} \backslash (\cup_{u \in V} X_u^k)$ be the updated residue returned by Core$(v, V \backslash D, R^{k - 1})$.
		Then $A_u^{K_v} = A_u^0 \cup (\cup_{i=1}^{K_v} X_u^i)$ denotes the pieces allocated to agent $u \in V$ and $R^{K_v}$ denotes the residue when the `while' loop terminates.
		Assume by induction that the allocation $\{A_u^0\}_{u \in V}$ is locally proportional and agents in $D$ are dominant, we prove that a) the allocation $\{A_u^{K_v}\}_{u \in V}$ is still locally proportional, and b) agent $v$ becomes dominant by setting an appropriate value for $K_v$.
		Note that the base case holds due to Theorem \ref{Theorem: the analysis of the CreateDominance Protocol}, by running Algorithm \ref{Algorithm: create dominance} before Algorithm \ref{Algorithm: diffuse dominance}.
		
		During the `while' loop, the residue is allocated among agents in $V \backslash D$ by the Core Protocol in an envy-free manner, thus agents in $V \backslash D$ remain proportional.
		Besides, since agents in $D$ are already dominant, they will remain proportional despite how the residue is allocated.
		As a result, the allocation $\{A_u^{K_v}\}_{u \in V}$ is still locally proportional.
		
		Let $|V \backslash D| = m$.
		We show that agent $v$ will be dominant by setting $K_v = (m - 2) \cdot \ln 3$ if $m \geq 3$ and $1$ otherwise.
		When $m = 2$, the Core Protocol reduces to the famous Cut and Choose Protocol where the residue is allocated completely among agents in $V \backslash D$.
		Since at this point the whole cake has been allocated, the protocol terminates with the desired properties.
		When $m \geq 3$, we first claim that agent $v$'s bonus over her neighbors is at least $\sum_{k=1}^{K_v} f_v(X_v^k)$.
		Since agent $v$ is a neighbor of $D$, there exists an agent $r' \in D$ such that they are adjacent.
		For each snapshot $\{X_u^k\}_{u \in V}$, the piece $X_{r'}^k$ is empty since agent $r'$ was excluded from the allocation of $R^{k - 1}$.
		Hence agent $v$'s bonus over agent $r'$ in this snapshot is $f_v(X_v^k)$.
		The claim follows by the assumption that $\{A_u^0\}_{u \in V}$ is locally proportional.
		Next we show that the value $\sum_{k=1}^{K_v} f_v(X_v^k)$ is larger than the value $f_v(R^{K_v})$, which implies that agent $v$ is dominant.
		By the properties of the Core Protocol, we have $f_v(X_v^k) = \frac{1}{m} f_v(R^{k-1})$ and $f_v(R^k) \leq \frac{m-2}{m} f_v(R^{k-1})$. Therefore,
		\begin{eqnarray*}
			\sum_{k=1}^{K_v}f_v(X_v^k) &=& \frac{1}{m}\sum_{k=0}^{K_v-1}f_v(R^k) \\
			&\geq& \frac{1}{m}\sum_{k=0}^{K_v-1}\left(\frac{m}{m-2}\right)^{K_v-k}f_v(R^{K_v}) \\
			&=& \frac{1}{2}\left(\left(\frac{m}{m-2}\right)^{K_v}-1\right)f_v(R^{K_v}) \\
			&\geq& f_v(R^{K_v}).
		\end{eqnarray*}
		The last inequality holds since we choose $K_v = (m - 2) \cdot \ln 3$.

		Finally, since there are $n$ agents and each agent invokes the Core Protocol $K_v = O(n)$ times, the protocol only needs $O(n^2)$ calls of the Core Protocol.
	\end{proof}

    \subsection{The Complexity of Core Protocol}
    \label{Subsection: the complexity of core protocol}

	To further reduce the overall complexity of our protocol, we give a tighter upper bound of the Core Protocol in this subsection.
	Previously, Aziz and Mackenzie \cite{AzizM16-Discrete-Any} showed that the Core Protocol is bounded by $n^{2n + 3}$ queries.
	They did not optimize their analysis since the main focus of their paper is to obtain a \emph{bounded} protocol.
	For readability of our analysis, we present the Core Protocol (Algorithm \ref{Algorithm: core protocol}) as well as the SubCore Protocol (Algorithm \ref{Algorithm:SubCore Protocol}) it invokes in this subsection.

	\begin{lemma}
		\label{Lemma: complexity of core protocol}
		The Core Protocol uses $O(c^n)$ queries for any constant $c > 4$.
	\end{lemma}

    \begin{proof}
    	We directly analyze the complexity of the SubCore Protocol, since the Core Protcol needs only $n - 1$ additional queries.
    	Denote by $T(n'', n')$ the number of queries that the SubCore Protocol needs to allocate $n''$ pieces among $n'$ agents ($n'' \geq n'$).
    	The `for' loop iterates for $m = 1, \cdots, n'$.
    	In the $m$-th iteration, agent $m$ needs $n''$ evaluations in step 2.
    	In step 5, each agent needs $n'' - (m - 1)$ evaluations on uncontested pieces and $m - 1$ cuts on contested pieces.
    	Thus this step costs at most $n'n''$ queries in total.
    	The `while' loop may iterate for $|W| = 1, \cdots, m - 2$ in the worst case.
    	In each iteration, the SubCore Protocol is invoked in step 9 on the $m - 1$ contested pieces among agents in $W$ and therefore requires $T(m - 1, |W|)$ queries. Step 11 costs $|W| + 1$ evaluations.
    	Step 13 costs another $T(m - 1, m - 1)$ queries.
    	To summarize, the number of queries used in the $m$-th iteration of the `for' loop is at most
    	\[ \sum_{i = 1}^{m - 1} T(m - 1, i) + \sum_{i = 1}^{m - 2} (i + 1) + n'' + n'n''. \]
    	Therefore, $T(n'', n')$ satisfies the recurrence
    	\begin{eqnarray*}
    		T(n'', n') & \leq & \sum_{m = 1}^{n'} \left(\sum_{i = 1}^{m - 1} T(m - 1, i) + {m - 1 \choose 2} + (m - 2) + n'' + n'n''\right) \\
    		&=& \sum_{i = 1}^{n'} \sum_{m = i + 1}^{n'} T(m - 1, i) + {n' \choose 3} + \frac{n'(n' - 3)}{2} + n'n'' + n'^2n'' \\
    		&\leq& \sum_{i = 1}^{n'} (n' - i) T(n' - 1, i) + {n' \choose 3} + \frac{n'(n' - 3)}{2} + n'n'' + n'^2n'' \\
    		&=& \sum_{i = 1}^{n' - 1} i T(n' - 1, n' - i) + {n' \choose 3} + \frac{n'(n' - 3)}{2} + n'n'' + n'^2n''.
    	\end{eqnarray*}
    	We then prove by induction on $n' + n''$ that $T(n'', n') \leq \tilde{c}^{n' + n''}$ for any constant $\tilde{c} > 2$ and sufficiently large $n' + n''$ such that
    	\[ {n' \choose 3} + \frac{n'(n' - 3)}{2} + n'n'' + n'^2n'' \leq \left(1 - \frac{1}{(\tilde{c} - 1)^2}\right) \tilde{c}^{n' + n''}. \]
    	Under the above assumptions and the fact that $n'' \geq n'$, we have
    	\begin{eqnarray*}
    		T(n'', n') & \leq & \sum_{i = 1}^{n' - 1} i \tilde{c}^{2n' - i - 1} + {n' \choose 3} + \frac{n'(n' - 3)}{2} + n'n'' + n'^2n'' \\
    		&\leq& \tilde{c}^{n' + n''} \sum_{i = 1}^{n' - 1} i \tilde{c}^{- i - 1} + \left(1 - \frac{1}{(\tilde{c} - 1)^2}\right) \tilde{c}^{n' + n''} \\
    		&\leq& \frac{1}{(\tilde{c} - 1)^2} \tilde{c}^{n' + n''} + \left(1 - \frac{1}{(\tilde{c} - 1)^2}\right) \tilde{c}^{n' + n''} \\
    		&=& \tilde{c}^{n' + n''}.
    	\end{eqnarray*}
    	Therefore, the Core Protocol is bounded by $n - 1 + T(n, n - 1) = O(\tilde{c}^{2n})$ queries.
    \end{proof}

    \begin{algorithm}
    	\SetAlgoLined
    	\KwIn{Specified cutter (say agent $i \in N$), agent set $N$ such that $i \in N$, and unallocated cake $R$.}
    	\KwOut{An envy-free allocation of cake $R' \subset R$ for agents in N and updated unallocated cake $R \backslash R'$.}
    	
    	Ask agent $i$ to cut the cake $R$ into $n$ equally preferred pieces.
    	
    	Run SubCore Protocol on the $n$ pieces with agent set $N \backslash \{i\}$ with each agent having a benchmark value as zero.
    	The call gives an allocation to the agents in $N \backslash \{i\}$ such that one of the $n$ pieces is untrimmed and unallocated.
    	
    	Give $i$ one of the unallocated untrimmed pieces from the previous step.
    	
    	\Return envy-free partial allocation (in which each agent gets a connected piece) as well as the unallocated cake.
    	
    	\caption{Core Protocol}
    	\label{Algorithm: core protocol}
    \end{algorithm}
    
    \begin{algorithm}
    	\SetAlgoLined
    	\KwIn{Cake cut into $n''$ pieces (with $ n'' \leq n $) to be allocated among agents in set $\{1, \cdots, n'\} = N' \subseteq N$ with $n' = |N'|$ and a benchmark value $b_j$ for each $ j \in N' $.
    		\{We only call SubCore if the benchmark values are such that there exists an envy-free allocation of agents in $N'$ where each agent gets at most one of the pieces giving him at least the specified benchmark value.\} }
    	\KwOut{A neat envy-free partial allocation for agents in $N'$ in which each agent $j \in N'$ gets a connected piece of value at least $b_j$.}
    	
    	\For{$m = 1$ to $n'$}{
    		\uIf{the piece agent $m$ preferred at the launch of the protocol is still unallocated}{
    			we tentatively give agent $m$ that piece and go to the next iteration of the `for' loop.
    		}
    		\Else{
    			the first $m$ agents are contesting for the same $m - 1$ tentatively allocated pieces. We call them the $contested$ pieces.
    			For each agent $j$, set $b_j'$ to be the maximum of $b_j$ and agent $j$'s value of the most preferred uncontested piece.
    			Then each agent in $[m]$ is asked to place a trim on all contested pieces of high enough value so that the contested piece on the right hand side of her trim is of the same value as $b'_j$.
    			
    			Set $W$ to be the set of agents who trimmed most (had the rightmost trim) in some piece.
    			
    			\While{ $|W| < m - 1$}{
    				Ignore the previous trims of agents in $W$ from now on and forget the previous allocation.
    				
    				Run SubCore Protocol on the contested pieces with $W$ as the target set of agents (with $b'_j$ as their benchmark value input) and for each contested piece, the part to the left side of the rightmost trim by an agent in $[m] \backslash W$ is ignored.
    				\{The result of the recursive call of SubCore is an allocation that gives a (partial) contested piece to each of the agent in $W$.\}
    				
    				Take any unallocated contested piece $a$.
    				The current left margin (beyond which the piece is ignored) is by agent $i \in [m] \backslash W$.
    				\[ W \leftarrow W \cup \{i\}.\]
    				At this point the current allocation of agents in $W$ is tentative and not permanently made.
    				\{An agent from $[m] \backslash W$ has been added to $W$.
    				For the updated $W$, each agent in $W$ gets a partial contested piece and no agent envies an unallocated piece. Recall that for each piece, the left side of the rightmost trim by an agent in $[m] \backslash \ W$ is ignored.\}
    				
    				Update the value $b'_j$ of each agent $j$ in $W$ to equal the value of the piece that they have been tentatively allocated.
    			}
    			
    			Run SubCore on all agents in $W$ and the set of contested pieces, where we ignore the part to the left of the trim make by the agent in $[m] \backslash W$.
    			The benchmark of each $j \in W$ is $b'_j$.
%    			\{At this point $|W| = m - 1$ and each agent in $W$ has a tentatively allocated contested piece.\}
    			
    			The only agent $j$ remaining in $[m] \backslash W$ is tentatively given her most preferred uncontested piece.
    		}
    	}
    	
    	\Return envy-free partial cake for agents in $N'$  as well as the unallocated cake.
    	% (such that each agent gets a connected piece that is on the right hand side of the original piece she trimmed most)
    	
    	\caption{SubCore Protocol}
    	\label{Algorithm:SubCore Protocol}
    \end{algorithm}

    \subsection{Analysis of Main Protocol}
    \label{Subsection: analysis of main protocol}

	Combining all the results from previous subsections, we have
	\begin{theorem}
		Algorithm \ref{Algorithm : main protocol} returns a locally proportional allocation on graph $G$ with $O(14^n)$ queries.
	\end{theorem}
	\begin{proof}
		By Theorem \ref{Theorem: the analysis of the CreateDominance Protocol} and \ref{Theorem: the analysis of the DiffuseDominance Protocol}, the allocation returned by the Main Protocol is locally proportional.
		The CreateDominance Protocol is dominant in the overall complexity and uses $O(n \cdot \left( \frac{6n}{k} \right) ^ k T(n))$ queries by Theorem \ref{Theorem: the analysis of the CreateDominance Protocol}, where $k$ is the length of the shortest path $P$ found in it and $T(n)$ is the number of queries of the Core Protocol.
		Since the Main Protocol selects an agent $r$ from the center of $G$, we have $k \leq n / 2$.
		By setting $c = 7 / \sqrt{3} > 4$ in Lemma \ref{Lemma: complexity of core protocol}, $n T(n) = O((7 / \sqrt{3})^n)$.
		Hence the overall complexity of the Main Protocol is bounded by $O((2 \sqrt{3})^n \cdot (7 / \sqrt{3})^n) = O(14^n)$.
	\end{proof}

    We remark that when the underlying graph $G$ is a path, the CreateDominance Protocol may invoke the Core Protocol roughly $n \cdot 5^{n / 2}$ times in the worst case.
    Thus even if the Core Protocol used polynomial queries, Algorithm \ref{Algorithm : main protocol} is still exponential.

	\section{Towards Envy-freeness on Trees}
	\label{Secton: toward envy-freeness on trees}

	This section briefly discusses locally envy-free allocations on trees.
	Previously, a \emph{continuous} locally envy-free protocol on trees was already known \cite{BeiQZ17-Networked}.
	However, despite of the globally envy-free protocol designed in \cite{AzizM16-Discrete-Any}, a \emph{simple} discrete and bounded one on trees is still unknown.
	As we have seen, the Core Protocol in \cite{AzizM16-Discrete-Any} which obtains an envy-free partial allocation is the cornerstone of several cake cutting protocols.
	In light of this, we present in this section a slightly weaker ``Core Protocol'' (Algorithm \ref{Algorithm: tree core}) which specializes on trees, but with a significantly improved complexity of $O(n^2)$ queries.

	\begin{algorithm}
		\SetAlgoLined
		\KwIn{A rooted tree $T$ with root $r$, $|T| = n$, and the cake $[0, 1]$.}
		\KwOut{A locally envy-free partial allocation on $T$ and a residue $R$.}

		Root $r$ cuts the cake into $n$ equally preferred pieces in her own measure and temporarily takes all these $n$ pieces.

		Initialize the residue $R = \emptyset$.

		\For{each agent $v \in T$  in a BFS order }{

			Agent $v$ have received $|T(v)|$ pieces at the moment, where $T(v)$ denotes the subtree rooted at $v$.
			Among them agent $v$'s least favorite piece is denoted by $p^*$.
			For each piece $p$ she holds, agent $v$ is asked to cut a part of it such that she thinks the remaining part $p'$ and $p^*$ are of the same value.
			Piece $p - p'$ is added into the residue $R$.

			\For{each immediate child $u$ of $v$ in an arbitrary order}{

				Agent $u$ takes $|T(u)|$ pieces she values the highest from the remaining pieces that agent $v$ holds.

			}

			The last piece that agent $v$ currently holds is allocated to her.

		}

		\Return the current allocation and the residue.

		\caption{TreeCore}
		\label{Algorithm: tree core}
	\end{algorithm}

	\begin{theorem}
		Algorithm \ref{Algorithm: tree core} returns a locally envy-free partial allocation on tree $T$ with $2n^2$ queries.
		Indeed, an agent never envies her descendants on $T$.
		Besides, root $r$ of $T$ obtains $1 / n$ of the cake.
	\end{theorem}

	\begin{proof}
		Step 5 to 8 are well-defined since $|T(v)| = \sum_{u:\mbox{ }child\mbox{ }of\mbox{ }v} |T(u)| + 1$.
		It follows that each agent gets possibly a part of exactly one piece cut by root $r$ after the protocol terminates.
		Then root $r$ obtains $1 / n$ of the cake since she cuts the cake into $n$ equal pieces and reserves one of them.
		Next, we show that no envy happens when agent $v$'s children take pieces from him.
		For a child $u$ of agent $v$, she takes $|T(u)|$ pieces which she values the highest from the remaining pieces that agent $v$ holds when it is agent $u$'s turn.
		Agent $u$ will finally obtain a piece that is of the same value as the least valuable one among those $|T(u)|$ pieces.
		Despite that, it is still better than the piece allocated to agent $v$, which guarantees agent $u$ will not envy agent $v$.
		As for agent $v$, since she cuts all $|T(v)|$ pieces she holds into an equal value before her children take pieces from him, she will not envy her children despite whichever piece she is finally allocated.
		In addition, $v$ will not envy her descendants since the pieces only gets smaller as the protocol proceeds.

		The root needs $(n-1)$ cut queries to obtain $n$ equal pieces.
		Each non-root agent $u$ needs to evaluate at most $|T(v)|$ pieces held by her parent $v$ and selects $|T(u)|$ pieces from them.
		She will make $|T(u)| - 1$ cuts on her selected pieces.
		As a result, each agent needs $|T(v)| + |T(u)| - 1 \leq 2n$ queries and hence the protocol is bounded by $2n^2$.
	\end{proof}

	\section{Conclusion}
	\label{Section: conclusion}

	There remains several open problems on the cake cutting problem under social graphs.
	First of all, it remains unknown whether or not a locally proportional allocation on any undirected graph can be found using polynomial queries.
	Next, does there exist a \emph{simple} locally proportional protocol for any given \emph{directed} graph?
	The DiffuseDominance Protocol still works for directed graphs.
	However, the exchange operation in the CreateDominance Protocol fails, since it requires that agents who are prepared to exchange their pieces are able to see each other.
	Finally, a simple and discrete locally envy-free protocol still lacks even for special classes of graphs.

%
%	\appendix
%	\appendixpage
%
%	\section{The Core Protocol}
%
%	For completeness, we present the Core Protocol in \cite{AzizM16-Discrete-Any} as Algorithm \ref{Algorithm: core protocol}.
%	It invokes the SubCore Protocol (Algorithm \ref{Algorithm:SubCore Protocol}).

	\newpage
	\bibliographystyle{plain}
	\bibliography{CakeCutting_arXiv}
\end{document}